\documentclass[conference]{IEEEtran}

\usepackage{amsmath,amssymb,amsthm,mathtools}
\usepackage{bm,bbm}
\usepackage{graphicx}
\usepackage{subcaption}
\usepackage{booktabs}
\usepackage{xcolor}
\let\labelindent\relax
\usepackage{enumitem}
\usepackage[hidelinks]{hyperref}
\usepackage[capitalize,noabbrev]{cleveref}
\usepackage{placeins}
\usepackage{microtype}
\usepackage{algorithm}
\usepackage{algorithmic}
\usepackage{tikz}
\usetikzlibrary{positioning,arrows.meta,calc}

\newtheorem{proposition}{Proposition}

\theoremstyle{definition}

\newtheorem{remark}{Remark}

\newcommand{\R}{\mathbb{R}}
\newcommand{\C}{\mathbb{C}}
\newcommand{\E}{\mathbb{E}}
\newcommand{\herm}{\mathsf{H}}
\newcommand{\trans}{\mathsf{T}}
\newcommand{\CN}{\mathcal{CN}}
\newcommand{\Pa}{\mathrm{Pa}}
\DeclareMathOperator{\tr}{tr}
\DeclareMathOperator{\diag}{diag}

\newcommand{\MI}{I}
\newcommand{\param}{\bm{\eta}}

\title{Mutual Information Optimization via K-Recursion
and Automatic Differentiation\\
for Linear Gaussian Wireless Networks}

\author{%
\IEEEauthorblockN{Tadashi~Wadayama and Na~Siqi}\\
\IEEEauthorblockA{Department of Computer Science,
Nagoya Institute of Technology, Nagoya 466-8555, Japan\\
Email: wadayama@nitech.ac.jp}%
}

\begin{document}
\maketitle

\begin{abstract}
We present a differentiable framework for end-to-end mutual
information (MI) optimization over linear Gaussian directed acyclic
graphs (DAGs). The framework targets network-wide design under global
constraints, such as a total transmit power budget, and covers MIMO
precoding, amplify-and-forward relays, RIS-aided channels, and
branching/merging topologies within a common linear Gaussian model.
Its core ingredient is a \emph{K-recursion} that analytically
propagates all node-pair covariances along the DAG in topological
order, including non-adjacent cross-covariances that are necessary
for correctly handling branching and merging paths. The resulting
covariances yield a closed-form log-determinant expression for the
end-to-end MI as a smooth function of the controllable factors.
Complex-valued reverse-mode automatic differentiation on this
K-recursion then returns the exact Wirtinger gradient at every
controllable factor in a single backward sweep, and projected gradient
ascent (PGA) is used to maximize the MI under the global constraints. Because
no closed-form gradient expression per topology is required, the same
topology-agnostic implementation applies to any linear Gaussian DAG.
A single topology-agnostic implementation is applied to four
representative DAG classes: single-link MIMO, a diamond DAG, a two-hop
AF relay, and input-covariance shaping. The same implementation
reaches the classical water-filling optimum in the settings where it
is available and yields MI improvements in non-single-link topologies
without using topology-specific gradient formulas. A further experiment on a multi-layer Gaussian
network (11 nodes, 5 layers) illustrates applicability to nontrivial
multi-layer topologies for which no closed-form gradient is available.
\end{abstract}


\section{Introduction}
\label{sec:intro}

Wireless networks~\cite{tse_viswanath2005} increasingly involve
cascaded, distributed, and reconfigurable signal transformations, as
in MIMO transceivers, multi-hop amplify-and-forward (AF) relay
networks, cell-free MIMO, reconfigurable intelligent surfaces (RIS),
and cooperative sensing arrays, and are no longer assemblies of
independently designed blocks.
End-to-end gains can often be obtained by network-wide joint
optimization of the controllable parameters, such as precoders,
power allocations, RIS phase profiles, and relay processing, under
global constraints such as a total transmit power budget, rather
than optimizing each block in isolation.
The mutual information (MI) between the transmitted signal and the
final observation captures this achievable performance in a
decoder-agnostic manner, and the Palomar--Verd\'u
gradient~\cite{palomar2006} established the analytical foundation for
MI-based design in single linear vector Gaussian channels two decades
ago. Revisiting this foundation today is worthwhile: what is still lacking
is a simple topology-agnostic mechanism that (i) extends to arbitrary
network topologies built from cascaded linear Gaussian stages and
(ii) is compatible with modern automatic differentiation (AD), so
that no new gradient formula needs to be derived for each new topology.

Such systems admit a clean abstraction as a
\emph{linear Gaussian directed acyclic graph} (Gaussian-DAG):
each node is a circular complex Gaussian vector, each edge is a
linear transformation that may include controllable parameters, and
additive Gaussian noise enters at every non-root node.
Multipath superposition, branching transmissions, multi-hop
processing, and structured beamforming can be represented within this
template under linear Gaussian modeling assumptions.
Existing approaches address parts of this problem but have different
limitations:
classical closed-form solutions (waterfilling~\cite{telatar1999},
the Palomar--Verd\'u single-link MIMO gradient~\cite{palomar2006},
and AF relay designs) are derived problem by problem and do not
extend to general DAG topologies;
variational and contrastive MI estimators (MINE~\cite{mine2018},
InfoNCE~\cite{infonce2018}) target estimation rather than design;
and the score-based information-gradient framework
of~\cite{wadayama2026dag} handles general stochastic DAGs but
requires denoising score matching (DSM) and Monte Carlo estimation.
For the practically important \emph{linear Gaussian} class, the
classical closed-forms are too narrow, while this setting permits
a simpler fully analytic construction than the score-based machinery.

This paper proposes a unified differentiable framework with exact MI
evaluation and exact gradient computation for end-to-end MI
optimization over linear Gaussian-DAGs.
The MI is evaluated in closed form as a log-determinant of
covariances, and these covariances are constructed analytically by
\emph{K-recursion} propagating all node-pair covariances along the
DAG in topological order, including the cross-covariances between
non-adjacent nodes that are essential when paths branch and merge.
The resulting forward computation graph is built entirely from
differentiable matrix primitives, so that complex-valued reverse-mode
AD returns the exact Wirtinger gradient at every controllable factor
in a single backward pass.
This gradient is the input required by projected gradient ascent
(PGA), which iteratively ascends the MI while enforcing the
network-wide constraints, such as a total transmit power budget, by
projecting each iterate onto the feasible set.

\subsection*{Contributions}
\begin{itemize}[leftmargin=1.5em,nosep]
\item A \textbf{unified Gaussian-DAG model} (\cref{sec:model}) that
covers single-link MIMO, multi-hop AF relay, and branching/merging
topologies under a common edge-factorization formalism, and can also
represent RIS-aided channels at the modeling level.
\item \textbf{K-recursion} (\cref{sec:recursion}): an analytical
recursion that constructs all node-pair covariances along the DAG in
topological order; the cross-covariances between non-adjacent nodes
are shown to be indispensable in general DAGs.
\item \textbf{Exact MI gradient and PGA optimization}
(\cref{sec:mi_grad,sec:pga}): the log-det MI is differentiated by
complex reverse-mode AD and optimized under network-wide
constraints by projected gradient ascent (PGA).
\item \textbf{Numerical experiments} (\cref{sec:experiments})
illustrating gradient correctness, recovery of classical results, and
the necessity of cross-covariance handling in branching/merging DAGs.
\end{itemize}

\subsection*{Code availability}
A reference PyTorch implementation accompanies this paper at
\url{https://github.com/wadayama/gaussian-dag} (MIT license).
The repository includes the K-recursion / log-det MI / PGA library,
runnable scripts that reproduce the four panels of
\cref{fig:pga_gallery} and the multi-layer experiment of
\cref{fig:random_network} exactly, and a five-part tutorial
walkthrough.

\section{Preliminaries}
\label{sec:preliminaries}

\subsection{Notation}
\label{sec:notation}
Uppercase italic letters (e.g., $X, V_j$) denote random variables or
random vectors, while boldface letters (e.g., $\bm{A}, \bm{\Sigma}$)
denote deterministic vectors and matrices.
$\bm{A}^{\trans}$ denotes the transpose;
$\bm{A}^{\herm}$ denotes the Hermitian (conjugate) transpose;
$\bm{A}^{*}$ denotes the entry-wise complex conjugate;
$\bm{\Sigma}\succ\bm{0}$ means Hermitian positive definite (HPD);
$\bm{\Sigma}\succeq\bm{0}$ means Hermitian positive semidefinite (PSD);
$\bm{I}_d$ is the $d\times d$ identity matrix;
$\tr(\cdot)$ and $\det(\cdot)$ are trace and determinant;
$\E[\cdot]$ is expectation;
$\Re(\cdot)$ takes the real part;
$\|\cdot\|_F$ is the Frobenius norm.
For a directed acyclic graph $\mathcal{G}=(\mathcal{V},\mathcal{E})$
on a node set $\mathcal{V}=\{V_1,\dots,V_M\}$,
$\Pa(j)\subset\{1,\dots,M\}$ denotes the parent \emph{index} set of
node $V_j$, i.e.,
$\Pa(j)=\{i:(V_i\!\to\!V_j)\in\mathcal{E}\}$.

\subsection{Circular Complex Gaussian and Log-Det MI}
\label{sec:cgauss}
A random vector $Y\in\C^d$ is \emph{circular complex Gaussian},
written $Y\sim\CN(\bm{0},\bm{\Sigma})$ with covariance
$\bm{\Sigma}\succ\bm{0}$, if its density is
$p_Y(\bm{y})=\det(\pi\bm{\Sigma})^{-1}\exp\bigl(-\bm{y}^{\herm}\bm{\Sigma}^{-1}\bm{y}\bigr)$.
For jointly circular complex Gaussian $(X,Y)$ with
$\bm{\Sigma}_X,\bm{\Sigma}_Y,\bm{\Sigma}_{Y|X}\succ\bm{0}$, the mutual
information~\cite{coverthomas2006} admits the log-determinant form
\begin{equation}
\MI(X;Y)=\log\det\bm{\Sigma}_Y-\log\det\bm{\Sigma}_{Y|X},
\label{eq:logdet_mi}
\end{equation}
where
\begin{equation}
\bm{\Sigma}_{Y|X}=\bm{\Sigma}_Y-\bm{\Sigma}_{YX}\bm{\Sigma}_X^{-1}\bm{\Sigma}_{XY}
\label{eq:schur}
\end{equation}
is the Schur-complement conditional covariance (the complex case
lacks the $1/2$ coefficient of the real counterpart). This
log-det form is at the heart of the present paper: once the
relevant covariances are constructed \emph{as differentiable
functions of design parameters}, the MI itself becomes a
differentiable scalar amenable to gradient-based optimization.

\subsection{Wirtinger Calculus}
\label{sec:wirtinger}
For a real-valued function $f(\bm{\Theta})\in\R$ of a complex matrix
variable $\bm{\Theta}\in\C^{p\times q}$, the Wirtinger partial
derivatives~\cite{schreier_scharf2010} satisfy
$\partial f/\partial\bm{\Theta}=(\partial f/\partial\bm{\Theta}^{*})^{*}$,
and the steepest-ascent direction in the standard real-Euclidean
metric is the conjugate-side derivative
\begin{equation}
\nabla_{\bm{\Theta}^{*}}f\triangleq\left(\frac{\partial f}{\partial\bm{\Theta}^{*}}\right)^{\!\trans}.
\label{eq:wirtinger_grad}
\end{equation}
Throughout the paper ``gradient'' and ``$\nabla_{\bm{\Theta}^{*}}$''
refer to this textbook conjugate-side derivative, i.e., the
real-Euclidean ascent direction represented in complex coordinates;
this is the quantity obtained exactly by reverse-mode AD on the
K-recursion (\cref{sec:mi_grad}).\footnote{%
Under the Wirtinger convention adopted here,
PyTorch~\cite{paszke2019pytorch} populates each complex leaf's
\texttt{.grad} attribute with
$2\,(\partial\mathcal{L}/\partial\bm{\Theta}^{*})^{\trans}$; the
factor of two is absorbed into the step size.}

\subsection{Related Work}
\label{sec:related}
MI-based system design for vector Gaussian channels traces back to
Palomar--Verd\'u~\cite{palomar2006}, whose closed-form precoder
gradient $\nabla_{\bm{F}}\MI$ is the conceptual origin of this work.
The present work differs in that the covariance recursion itself is
used as the differentiable computational graph, so a single
reverse-mode AD sweep replaces per-topology derivations.
Water-filling~\cite{telatar1999} and 2-hop AF relay solutions are
similarly problem-specific and do not extend automatically. A separate
line of work estimates MI from samples via variational bounds
(MINE \cite{mine2018}, InfoNCE \cite{infonce2018}) or
nonparametric estimators; these introduce bias and variance that
are difficult to control when used as a design tool, and they do
not exploit the analytic tractability available in the linear
Gaussian regime. The recent score-based information-gradient
framework~\cite{wadayama2026dag} provides an exact gradient identity
for \emph{general} stochastic DAGs, implemented through denoising
score matching (DSM) and Monte Carlo estimation; the present paper
can be viewed as its analytical specialization for the linear Gaussian
case, where the score functions are analytically known and the entire
MI gradient reduces to the closed-form construction developed here,
with no DSM training, score network, or Monte Carlo sampling
required.

\section{Linear Gaussian-DAG Model}
\label{sec:model}

This section formalizes the linear Gaussian DAG model, with the
diamond of \cref{fig:dag_intro} as a running example: we establish
notation for nodes and edges
(\cref{sec:nodes_edges}), factorize each edge matrix into
controllable and constant pieces (\cref{sec:edge_factor}), and
introduce the parameter and feasible sets that the optimization
will act on (\cref{sec:param_feasible}).

\begin{figure}[t]
\centering
\includegraphics[width=0.55\linewidth]{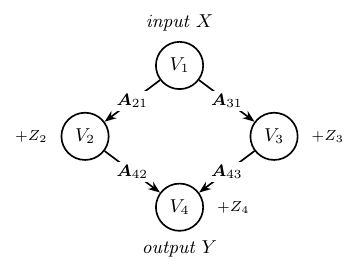}
\caption{An example of a linear Gaussian DAG (diamond network) with
input $X\!=\!V_1$, output $Y\!=\!V_4$, and edge transforms
$\bm{A}_{ji}$ (panel (b) of \cref{fig:topology}).}
\label{fig:dag_intro}
\end{figure}

\subsection{Nodes, Edges, and Input/Output}
\label{sec:nodes_edges}
Let $\mathcal{G}=(\mathcal{V},\mathcal{E})$ be a directed acyclic
graph with topologically ordered nodes
$\mathcal{V}=\{V_1,\dots,V_M\}$, where each node is a complex random
vector $V_j\in\C^{d_j}$. Here \emph{topologically ordered} means
that the index labelling satisfies $i<j$ for every edge
$(i\!\to\!j)\in\mathcal{E}$, equivalently
$\Pa(j)\subset\{1,\dots,j-1\}$ for all $j\geq 2$; such a labelling
exists because $\mathcal{G}$ is acyclic, and the K-recursion of
\cref{sec:recursion} processes the nodes along this order.
For exposition, we treat the input as the unique root
$X\triangleq V_1$ and the output as the terminal node
$Y\triangleq V_M$; multi-input or multi-output models
(broadcast/multiple-access/interference channels) are recovered by
stacking the appropriate root or sink nodes, and we omit this
straightforward extension here.
The input distribution is $X\sim\CN(\bm{0},\bm{\Sigma}_X)$ with
$\bm{\Sigma}_X\succ\bm{0}$, the HPD assumption ensuring that the
conditional log-determinant in the log-det MI of \cref{sec:cgauss}
is well-defined.
For each non-root node $j\in\{2,\dots,M\}$,
\begin{equation}
V_j=\sum_{i\in\Pa(j)}\bm{A}_{ji}V_i+Z_j,\quad
Z_j\sim\CN(\bm{0},\bm{\Sigma}_j),
\label{eq:node_eq}
\end{equation}
where $\bm{A}_{ji}\in\C^{d_j\times d_i}$ is the linear edge transform
from parent $i$ to node $j$, $\bm{\Sigma}_j\succeq\bm{0}$ is a fixed
PSD noise covariance, and the noise vectors
$\{Z_j\}_{j=2}^{M}$ are mutually independent and independent of $X$.
We adopt the standard zero-mean convention as is typical in
information-theoretic wireless modeling. The PSD assumption
$\bm{\Sigma}_j\succeq\bm{0}$ allows noiseless or rank-deficient
components at intermediate nodes; positive definiteness of the final
covariances $\bm{\Sigma}_Y, \bm{\Sigma}_{Y|X}$ is imposed only when
the log-det MI of \cref{sec:cgauss} is evaluated
(\cref{prop:MI_from_K}).
For the diamond example of \cref{fig:dag_intro}, $M = 4$ and
\eqref{eq:node_eq} specializes to
$V_2 = \bm{A}_{21} V_1 + Z_2$,
$V_3 = \bm{A}_{31} V_1 + Z_3$, and
$V_4 = \bm{A}_{42} V_2 + \bm{A}_{43} V_3 + Z_4$,
illustrating the branching at $V_1$ and the merging at $V_4$.

\subsection{Edge Factorization and Controllable Factors}
\label{sec:edge_factor}

\begin{figure*}[!t]
\centering
\includegraphics[width=0.85\textwidth]{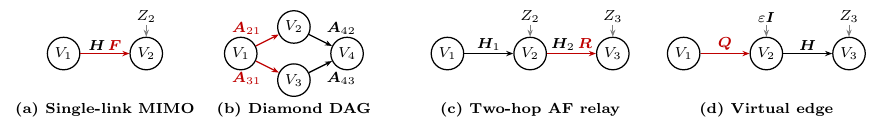}
\caption{Four representative linear Gaussian DAGs covered by this
framework, revisited in the numerical gallery of
\cref{sec:experiments}: (a) single-link MIMO, (b) diamond, (c) two-hop
AF relay, (d) virtual-edge input shaping. Edges containing a
controllable factor are drawn in
\textcolor{red!75!black}{red}, with the controllable factor itself
typeset in \textcolor{red!75!black}{red} and fixed channel factors
left in black; purely fixed edges are entirely black; dashed gray
arrows are additive Gaussian noise injections.}
\label{fig:topology}
\end{figure*}

Each edge matrix is decomposed into a product of $K_{ji}\geq 1$ factors,
\begin{equation}
\bm{A}_{ji}=\bm{A}_{ji}^{(1)}\bm{A}_{ji}^{(2)}\cdots\bm{A}_{ji}^{(K_{ji})},
\label{eq:edge_factor}
\end{equation}
read from output side ($\ell=1$) to input side ($\ell=K_{ji}$).
Each factor $\bm{A}_{ji}^{(\ell)}\in\C^{d_{ji}^{(\ell-1)}\times d_{ji}^{(\ell)}}$
has inner dimensions $\{d_{ji}^{(\ell)}\}_{\ell=0}^{K_{ji}}$ with
endpoints $d_{ji}^{(0)}=d_j$ and $d_{ji}^{(K_{ji})}=d_i$, so the
product in \eqref{eq:edge_factor} maps $\C^{d_i}$ to $\C^{d_j}$ as
required.
Each $\bm{A}_{ji}^{(\ell)}$ is labeled either \emph{controllable},
subject to design (e.g.,
precoders, power allocations, RIS phase profiles, relay gains), or
\emph{constant}, a fixed physical quantity (e.g., the channel
realization); we collect the index set of controllable factors as
\begin{equation}
\mathcal{C}\triangleq
\bigl\{(j,i,\ell) : \bm{A}_{ji}^{(\ell)}\ \text{is controllable}\bigr\}.
\label{eq:C_def}
\end{equation}
This factorization is flexible enough to encompass several canonical
wireless architectures~\cite{elgamal_kim2011} uniformly, as
illustrated in \cref{fig:topology}:
\begin{itemize}[leftmargin=1.5em,nosep]
\item\emph{Single-link MIMO} (\cref{fig:topology}(a);
$V_1\!\to\!V_2$, $V_1=X$, $V_2=Y$):
$\bm{A}_{21}=\bm{H}\bm{F}$, with constant channel $\bm{H}$ and
controllable precoder $\bm{F}$~\cite{palomar2006} (an additional
controllable combiner can be absorbed as a third factor in
\eqref{eq:edge_factor}).
\item\emph{Diamond / branching DAG}
(\cref{fig:topology}(b), detailed in \cref{fig:dag_intro}):
multiple branching paths from $X$ that merge at $Y$, each carrying
its own constant channel and controllable processing.
\item\emph{Multi-hop AF relay} (\cref{fig:topology}(c);
$V_1\!\to\!V_2\!\to\!\cdots\!\to\!V_M$):
the source-to-relay hop carries $\bm{A}_{21}=\bm{H}_1$, and each
subsequent relay-processing hop can be written as
$\bm{A}_{j+1,j}=\bm{H}_j\bm{R}_j$ with constant channel $\bm{H}_j$
and controllable relay gain $\bm{R}_j$.
\item\emph{Virtual-edge input shaping} (\cref{fig:topology}(d)):
adding a virtual source node and an auxiliary controllable edge
recasts input-covariance design ($\bm{\Sigma}_X=\bm{Q}\bm{Q}^{\herm}$
with controllable $\bm{Q}$) as ordinary edge-factor optimization, so
that input-covariance optimization is naturally subsumed by the same
K-recursion + PGA pipeline (\cref{rem:input_cov}).
\end{itemize}

\subsection{Parameter Set and Feasible Set}
\label{sec:param_feasible}
The optimization variable is the collection of controllable factors,
\begin{equation}
\param\triangleq
\bigl\{\bm{A}_{ji}^{(\ell)} : (j,i,\ell)\in\mathcal{C}\bigr\}.
\label{eq:eta_def}
\end{equation}
By default, each $\bm{A}_{ji}^{(\ell)}$ is treated as a free complex
matrix in autograd, and structural and power constraints are
expressed by a feasible set
$\mathcal{S}\subset\prod_{(j,i,\ell)\in\mathcal{C}}
\C^{d_{ji}^{(\ell-1)}\times d_{ji}^{(\ell)}}$
to which the iterates are projected (see \cref{sec:pga}).
Typical examples relevant in wireless design include:
\begin{itemize}[leftmargin=1.5em,nosep]
\item\emph{Total/per-factor Frobenius-norm budgets:}
\[
\textstyle
\sum_{(j,i,\ell)}\|\bm{A}_{ji}^{(\ell)}\|_F^2\le P
\quad\text{or}\quad
\|\bm{A}_{ji}^{(\ell)}\|_F^2\le P_{ji}^{(\ell)}.
\]
These include standard transmit-power constraints for source
precoders and serve as simple design-energy constraints for generic
controllable factors.
\item\emph{Diagonal/per-stream allocation},
$\bm{A}_{ji}^{(\ell)}=\diag(s_1,\dots,s_d)$ with $\bm{s}\in\C^d$.
\item\emph{Unit-modulus diagonal control},
$\bm{A}_{ji}^{(\ell)}=\diag(\theta_1,\dots,\theta_d)$ with
$|\theta_m|=1$, as in idealized RIS phase profiles.
\item\emph{Unit-gain scalar control},
$\bm{A}_{ji}^{(\ell)}=\alpha\bm{I}_d$ with $\alpha\in\C$.
\end{itemize}
Diagonal, scalar-gain, and unit-modulus structures are realized by
directly parameterizing the underlying scalars/vectors as autograd
leaves: the structural constraint is enforced by parameterization,
and any additional norm or modulus constraint is then projected in
this lower-dimensional parameter space, which keeps the gradient
computation transparent without changing the abstract framework.

The same parameterization mechanism also accommodates
\emph{parameter sharing}: multiple positions in $\bm{\eta}$ can be
tied to a single underlying autograd leaf, identifying controllable
factors that share a physical origin. A canonical example is a relay
node $i$ whose physical processing is applied \emph{once} and then
broadcast over all outgoing edges; this is expressed by writing
$\bm{A}_{ji}=\bm{H}_{ji}\bm{F}_i$ for every $j$ with $i\in\Pa(j)$,
so that the same controllable factor $\bm{F}_i$ appears in
$|\{j:i\in\Pa(j)\}|$ edge factorizations simultaneously. Reverse-mode
AD on the K-recursion handles such sharing transparently: gradient
contributions through every edge using $\bm{F}_i$ are accumulated by
the chain rule into a single $\nabla_{\bm{F}_i^{*}}\MI$ in one
backward sweep. The multi-layer experiment of
\cref{sec:exp_random_network} is an instance of this design.

In the experiments of \cref{sec:experiments}, we focus on
Frobenius-norm budgets on the controllable factors. Physical
node-transmit-power constraints, which may depend on the propagated
node covariances (e.g., the AF-relay transmit power
$\E\|\bm{R}V_j\|^2$), can also be incorporated through the same
K-recursion but may require different projection or penalty steps;
we leave such variants to future work.

\section{K-Recursion}
\label{sec:recursion}

This section presents the core analytic ingredient of the proposed
framework: a topological-order recursion that constructs all
\emph{node-pair covariances} of a linear Gaussian DAG as
differentiable functions of the controllable factors.
A by-product of this recursion is the closed-form expression of
$\bm{\Sigma}_Y$ and $\bm{\Sigma}_{Y|X}$ used in the log-det MI of
\cref{sec:cgauss}.

\subsection{Node-Pair Covariances and Storage Convention}
\label{sec:K_def}
For arbitrary node indices $(j,k)$ (not necessarily adjacent in the
DAG), define the node-pair covariance
\begin{equation}
\bm{K}_{jk}\triangleq\E\bigl[V_j V_k^{\herm}\bigr]\in\C^{d_j\times d_k}.
\label{eq:K_def}
\end{equation}
Hermitian symmetry $\bm{K}_{kj}=\bm{K}_{jk}^{\herm}$ implies that only
the canonical half $\{\bm{K}_{jk}: j\geq k\}$ needs to be stored.
Wherever an index reversal occurs in subsequent formulas, we adopt
the convention
\begin{equation}
\bm{K}_{ab}=\bm{K}_{ba}^{\herm}\quad\text{for}\ a<b
\label{eq:K_flip}
\end{equation}
to refer to the canonical block (\emph{Hermitian flip}).
This halves both the storage and the explicit computation.

\subsection{Recursion Formula}
\label{sec:K_recursion_prop}
The following proposition is the analytic core of this paper:
all node-pair covariances $\bm{K}_{jk}$ are obtained from
$\bm{\Sigma}_X$ and $\{\bm{\Sigma}_j\}$ by a finite sequence of matrix
products, sums, and Hermitian transposes, ordered along the DAG.

\begin{proposition}[K-recursion]
\label{prop:K_recursion}
Under the model of \cref{sec:nodes_edges} and the node equation
\eqref{eq:node_eq},
suppose that the nodes are processed in topological order and that
only the canonical covariance blocks $\bm{K}_{jk}$ with $j\geq k$ are
stored.
Whenever a non-canonical block $\bm{K}_{ab}$ with $a<b$ appears on
the right-hand side below, it is interpreted by the
\emph{Hermitian flip} of \eqref{eq:K_flip},
\begin{equation*}
\bm{K}_{ab}:=\bm{K}_{ba}^{\herm}\quad(a<b).
\end{equation*}
Then for all $j\geq k\geq 1$,
\begin{equation}
\bm{K}_{jk}=
\begin{cases}
\bm{\Sigma}_X & j=k=1,\\[2pt]
\displaystyle\sum_{i\in\Pa(j)}\!\!\bm{A}_{ji}\bm{K}_{ik}
  & j\geq 2,\,k<j,\\[6pt]
\displaystyle\sum_{i,i'\in\Pa(j)}\!\!\bm{A}_{ji}\bm{K}_{ii'}\bm{A}_{ji'}^{\herm}
  +\bm{\Sigma}_j & j\geq 2,\,k=j.
\end{cases}
\label{eq:K_recursion}
\end{equation}
\end{proposition}

\begin{proof}
The proof proceeds by strong induction on the topological order
$j=1,2,\dots,M$, with the inductive hypothesis that all canonical
K-blocks $\{\bm{K}_{j'k'}:j'\geq k',\,j'<j\}$ have already been
constructed. The two recursion branches in \eqref{eq:K_recursion}
together establish $\{\bm{K}_{jk}:k\leq j\}$ at step $j$, completing
the induction.

\emph{Base case ($j=k=1$):}
$V_1=X\sim\CN(\bm{0},\bm{\Sigma}_X)$ yields $\bm{K}_{11}=\bm{\Sigma}_X$.

\emph{Independence lemma:}
By induction on the topological order, for fixed edge factors $V_k$
is a deterministic linear function of $X$ and a subset of
$\{Z_l\}_{2\leq l\leq k}$. The case $k=1$ holds since $V_1=X$; for
$k\geq 2$, \eqref{eq:node_eq} with $\Pa(k)\subset\{1,\dots,k-1\}$
propagates the claim through the induction hypothesis. Since
$\{Z_l\}_{l=2}^{M}$ are mutually independent and independent of $X$,
$\E[Z_jV_k^{\herm}]=\bm{0}$ whenever $k<j$.

\emph{Cross-block ($j\geq 2$, $k<j$):}
Substituting \eqref{eq:node_eq} into \eqref{eq:K_def} gives
$\bm{K}_{jk}=\sum_{i\in\Pa(j)}\bm{A}_{ji}\E[V_iV_k^{\herm}]
+\E[Z_jV_k^{\herm}]$; the last term vanishes by the lemma. Each
$i\in\Pa(j)$ satisfies $i<j$, so $\E[V_iV_k^{\herm}]=\bm{K}_{ik}$
under the canonical-storage convention \eqref{eq:K_flip}, recovering
the cross-block branch.

\emph{Self-block ($j\geq 2$, $k=j$):}
Expanding both copies of $V_j$ via \eqref{eq:node_eq} and applying
the independence lemma to discard the $Z_j$ cross terms (each parent
$i,i'<j$) leaves
$\bm{K}_{jj}=\sum_{i,i'\in\Pa(j)}\bm{A}_{ji}\E[V_iV_{i'}^{\herm}]
\bm{A}_{ji'}^{\herm}+\bm{\Sigma}_j$; applying \eqref{eq:K_flip} to
each $\E[V_iV_{i'}^{\herm}]$ gives the self-block branch. The double
sum includes parent cross-covariance terms whenever
$|\Pa(j)|\geq 2$, which are generally nonzero in branching/merging
DAGs, as the diamond example \eqref{eq:K44_diamond} makes explicit.

\emph{Termination:}
Each block on the right-hand side of \eqref{eq:K_recursion} involves
indices strictly less than $j$ or pairs already in the canonical
store, so processing nodes in topological order yields all
$\{\bm{K}_{jk}:j\geq k\}$ in a single forward sweep.
\end{proof}

The recursion is executed in topological order $j=1,2,\dots,M$, and at
each step requires only previously computed K-blocks; the entire
collection $\{\bm{K}_{jk}:j\geq k\}$ is obtained in a single forward
pass.
Crucially, every operation in \eqref{eq:K_recursion} is a matrix
product, sum, or Hermitian transpose (all differentiable primitives
supported by complex autograd), so the recursion realizes a
\emph{differentiable forward computation graph} from the controllable
factors $\param$ to all $\bm{K}_{jk}$. Note that the second
double-sum in \eqref{eq:K_recursion} forces a merging node's
self-block to depend on the cross-covariance of its parents:
tracking only the auto-covariances $\bm{K}_{jj}$ is therefore
insufficient in general (see \cref{fig:K_dependency}, and the
diamond DAG of \cref{sec:experiments} for a numerical illustration).
Concretely, for the diamond DAG of \cref{fig:dag_intro} ($M=4$), the
merging self-block reads
\begin{equation}
\begin{aligned}
\bm{K}_{44}={}&\bm{A}_{42}\bm{K}_{22}\bm{A}_{42}^{\herm}
  +\bm{A}_{42}\bm{K}_{23}\bm{A}_{43}^{\herm}\\
  &+\bm{A}_{43}\bm{K}_{32}\bm{A}_{42}^{\herm}
  +\bm{A}_{43}\bm{K}_{33}\bm{A}_{43}^{\herm}+\bm{\Sigma}_4,
\end{aligned}
\label{eq:K44_diamond}
\end{equation}
with $\bm{K}_{23}=\bm{K}_{32}^{\herm}$ via the Hermitian flip; the
parent cross-block $\bm{K}_{32}$ explicitly enters $\bm{K}_{44}$ and
hence the log-det MI.

\begin{figure}[t]
\centering
\includegraphics[width=0.65\linewidth]{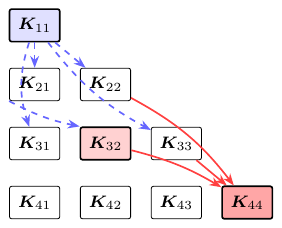}
\caption{K-block dependency for the diamond DAG
(\cref{fig:dag_intro}); only canonical blocks ($j\!\geq\!k$) are stored.
Blue dashed arrows: forward propagation from
$\bm{K}_{11}\!=\!\bm{\Sigma}_X$.
Red arrows: dependencies for the merging block $\bm{K}_{44}$, which
requires the parent cross-covariance $\bm{K}_{32}$ in addition to the
self-blocks $\bm{K}_{22},\bm{K}_{33}$.}
\label{fig:K_dependency}
\end{figure}

\begin{remark}[Computational complexity]
\label{rem:complexity}
Let $d_{\max}\!=\!\max_j d_j$ and $|\mathcal{E}|$ be the number of
DAG edges.
The cross-block and self-block updates of the K-recursion cost
$O(\sum_{j=2}^{M}(j\!-\!1)|\Pa(j)|d_{\max}^{3})
\subseteq O(M|\mathcal{E}|d_{\max}^{3})$
and $O(\sum_{j=2}^{M}|\Pa(j)|^{2}d_{\max}^{3})$, respectively.
For sparse DAGs with $|\Pa(j)|\!=\!O(1)$, the overall cost is
$O(M^{2}d_{\max}^{3})$, dominated by the cross-block updates,
and storing the canonical blocks requires $O(M^{2}d_{\max}^{2})$
memory.
Reverse-mode AD through this graph is of the same asymptotic order
but requires retaining the forward intermediates for the backward
pass.
\end{remark}

\subsection{Mutual Information from the K-Recursion}
\label{sec:output_cov}

With all node-pair covariances furnished by the K-recursion of
\cref{prop:K_recursion}, the mutual information itself admits a
closed-form expression as a smooth function of the controllable
factors $\param$. The following proposition formalizes this bridge
from the DAG model to a scalar differentiable objective, on which
the gradient-based optimization of \cref{sec:mi_grad,sec:pga} will
act.

\begin{proposition}[Mutual information from the K-recursion]
\label{prop:MI_from_K}
Consider the linear Gaussian DAG of \cref{sec:model}. Under the
regularity assumption
$\bm{\Sigma}_Y(\param), \bm{\Sigma}_{Y|X}(\param) \succ \bm{0}$,
the mutual information admits the closed form
\begin{equation}
\MI(X;Y)
=\log\det\bm{\Sigma}_Y(\param)-\log\det\bm{\Sigma}_{Y|X}(\param),
\label{eq:mi_logdet_param}
\end{equation}
where
\begin{equation}
\bm{\Sigma}_Y=\bm{K}_{MM},\qquad
\bm{\Sigma}_{Y|X}
=\bm{K}_{MM}-\bm{K}_{M1}\bm{K}_{11}^{-1}\bm{K}_{M1}^{\herm}
\label{eq:Sigma_Y_given_X}
\end{equation}
are obtained from the K-blocks of \cref{prop:K_recursion}. The map
$\param\mapsto\MI(X;Y)$ is smooth, as a real-valued function of the
real and imaginary parts of the controllable factors, on the open set
where the regularity assumption holds.
\end{proposition}

\begin{proof}
With the K-blocks computed, the quantities entering the log-det MI of
\cref{sec:cgauss} are read off directly: $\bm{\Sigma}_Y=\bm{K}_{MM}$
and $\bm{\Sigma}_{YX}=\bm{K}_{M1}$, with the index reversal handled by
\eqref{eq:K_flip}. The conditional output covariance follows from the
Schur-complement form of \cref{sec:cgauss} applied with these K-blocks.
The K-blocks are smooth in $\param$ by \cref{prop:K_recursion} (matrix
products and sums of controllable and constant factors), and
$\log\det$ is smooth on the open positive-definite cone; hence the
composition is smooth wherever
$\bm{\Sigma}_Y(\param),\bm{\Sigma}_{Y|X}(\param)\succ\bm{0}$.
\end{proof}

\begin{remark}[Effective-channel representation]
\label{rem:effective_channel}
For MI evaluation alone, the Gaussian-DAG can also be collapsed to
an equivalent linear Gaussian channel
$Y=\bm{G}_M X+R_M$, where $R_M$ is independent of $X$ and the
effective channel matrices satisfy $\bm{G}_1=\bm{I}_{d_X}$ and
$\bm{G}_j=\sum_{i\in\Pa(j)}\bm{A}_{ji}\bm{G}_i$ for $j\geq 2$. The
effective-noise covariance blocks
$\bm{C}_{jk}\triangleq\E[R_j R_k^{\herm}]$ obey the same recursion
as \cref{prop:K_recursion} but with the input covariance set to
zero, i.e.\ $\bm{C}_{11}=\bm{0}$, yielding
$\MI(X;Y)=\log\det(\bm{G}_M\bm{\Sigma}_X\bm{G}_M^{\herm}+\bm{C}_{MM})
-\log\det\bm{C}_{MM}$.
This is equivalent to the Schur-complement form
\eqref{eq:Sigma_Y_given_X}, since
$\bm{K}_{jk}=\bm{G}_j\bm{\Sigma}_X\bm{G}_k^{\herm}+\bm{C}_{jk}$.
Either representation is a valid \emph{differentiable computation
graph}, and reverse-mode AD on it (\cref{sec:autograd_grad}) yields
the exact MI gradient with respect to the controllable factors.
Although the effective-channel representation
$(\bm{G}_M,\bm{C}_{MM})$ is sufficient for a single source-sink MI,
the full K-recursion is the more general primitive: it exposes
intermediate covariances, parent cross-covariances at merging nodes,
and covariance quantities needed for internal constraints, multiple
sinks, or Bussgang-type linearizations of nonlinear elements in the
network (whose equivalent linear gains depend on intermediate
node-pair covariances).
\end{remark}

\section{Gradient Computation and Projected Ascent}
\label{sec:mi_grad}

With the closed-form mutual information of \cref{prop:MI_from_K} at
hand, we now compute its exact Wirtinger gradient by complex
reverse-mode AD and maximize $\MI(X;Y)$ under the network-wide
constraints by projected gradient ascent (PGA). The construction is
AD-framework agnostic; for concreteness, we use
PyTorch~\cite{paszke2019pytorch} as the AD engine throughout this
paper and flag the framework-specific conventions (e.g., the sign of
the returned gradient and the factor of two for complex leaves)
where they affect implementation.

\subsection{Reverse-Mode AD for the Wirtinger Gradient}
\label{sec:autograd_grad}
By \cref{prop:MI_from_K}, the map $\param \mapsto \MI(X;Y)$ factors
through the following five-stage composition, each stage consisting of
smooth operations on complex matrices:
\begin{enumerate}[label=(\arabic*),leftmargin=2.2em,nosep]
\item the controllable parameter set
$\param=\{\bm{A}_{ji}^{(\ell)}\}_{(j,i,\ell)\in\mathcal{C}}$;
\item edge matrices $\bm{A}_{ji}=\prod_{\ell}\bm{A}_{ji}^{(\ell)}$
(matrix products of controllable and constant factors);
\item the canonical K-blocks $\{\bm{K}_{jk}\}_{j\geq k}$ via
\eqref{eq:K_recursion} (sums of matrix products);
\item the output covariances
$\bm{\Sigma}_Y\!=\!\bm{K}_{MM}$, $\bm{\Sigma}_{YX}\!=\!\bm{K}_{M1}$,
and $\bm{\Sigma}_{Y|X}$ via the Schur complement
\eqref{eq:Sigma_Y_given_X};
\item the scalar real-valued loss $\MI(X;Y)$ via $\log\det(\cdot)$.
\end{enumerate}
The forward graph is built entirely from primitives standard in modern
automatic-differentiation engines (matrix product, sum, Hermitian
transpose, matrix inverse, log-determinant). Because the underlying
network is a DAG, the graph is itself acyclic, and reverse-mode AD
traverses it in a single well-defined backward sweep without unrolling
or fixed-point iteration. The DAG structure thus plays a dual role: it
makes the K-recursion well-defined in topological order
(\cref{prop:K_recursion}) and the AD backward sweep well-defined in
reverse topological order. Defining the loss
$\mathcal{L}(\param)\triangleq -\MI(X;Y)$, complex reverse-mode AD
returns the exact conjugate-side Wirtinger gradient
\begin{equation}
\nabla_{\bm{A}_{ji}^{(\ell)*}}\MI
\triangleq
\left(\frac{\partial\MI}{\partial\bm{A}_{ji}^{(\ell)*}}\right)^{\trans},
\qquad (j,i,\ell)\in\mathcal{C},
\label{eq:wirtinger_grad_factor}
\end{equation}
at every controllable factor in a single backward sweep through stages
(5)$\to$(1). In implementation we differentiate the loss
$\mathcal{L}=-\MI$; the ascent direction for $\MI$ is obtained by
changing the sign of the returned gradient, with the PyTorch
convention-dependent factor of two absorbed into the step size as
noted in \cref{sec:wirtinger}. The applicability of the Wirtinger
chain rule to arbitrary complex matrix parameters was formalized
in~\cite{schreier_scharf2010}; the reverse-mode AD employed here can
be viewed as its automatic, topology-agnostic execution. By the
cheap-gradient principle~\cite{baydin2018ad}, this backward sweep
costs at most a small constant multiple of the forward sweep, and
\emph{all} controllable-factor gradients are produced simultaneously.

In finite precision, $\bm{\Sigma}_Y$ and $\bm{\Sigma}_{Y|X}$ are
symmetrized as $\tfrac{1}{2}(\bm{\Sigma}+\bm{\Sigma}^{\herm})$ and
evaluated via $\log\det(\bm{\Sigma}+\epsilon\bm{I})=2\sum_i\log L_{ii}$
(Cholesky, small $\epsilon$) to keep the log-det well-conditioned and
smooth in $\param$.

\subsection{Projected Gradient Ascent}
\label{sec:pga}
The exact gradient $\nabla_{\param^{*}}\MI$ produced by reverse-mode
AD is precisely the input required by PGA:
\begin{equation}
\param^{(t+1)}
=\mathcal{P}_{\mathcal{S}}\bigl(\param^{(t)}
+\alpha_t\,\nabla_{\param^{*}}\MI(\param^{(t)})\bigr),
\label{eq:pga_update}
\end{equation}
where $\mathcal{P}_{\mathcal{S}}$ denotes the Euclidean projection
onto the feasible set $\mathcal{S}$,
\begin{equation}
\mathcal{P}_{\mathcal{S}}(\bm{\xi})
:=\arg\min_{\bm{\zeta}\in\mathcal{S}}\|\bm{\zeta}-\bm{\xi}\|_F^2,
\label{eq:projection}
\end{equation}
with $\mathcal{S}$ encoding the design constraints of
\cref{sec:param_feasible} and $\{\alpha_t\}$ a step-size schedule.
The update is interpreted factor-wise: each controllable factor
$\bm{A}_{ji}^{(\ell)}$ is independently ascended along its own
gradient before joint projection onto $\mathcal{S}$. Each iteration
costs one forward--backward sweep through the K-recursion for
$\nabla_{\param^{*}}\MI$, plus one application of
$\mathcal{P}_{\mathcal{S}}$ which is closed-form for the projection
types discussed below.
Because the objective \eqref{eq:mi_logdet_param} is generally
non-concave in $\param$, PGA provides only convergence to a KKT
stationary point under standard step-size schedules; multiple random
initializations are typically used in practice.

For the canonical wireless constraints introduced in
\cref{sec:param_feasible}, the projection $\mathcal{P}_{\mathcal{S}}$
admits the following closed forms:
\begin{itemize}[leftmargin=1.5em,nosep]
\item \emph{Total Frobenius budget}
$\sum_{(j,i,\ell)\in\mathcal{C}}\|\bm{A}_{ji}^{(\ell)}\|_F^2\le P$:
a \emph{single common} scale factor $s$ is applied to every
controllable factor,
\[
\bm{A}_{ji}^{(\ell)}\leftarrow s\,\bm{A}_{ji}^{(\ell)},\ \
s=\min\!\left\{1,\sqrt{\tfrac{P}{\sum_{(j,i,\ell)\in\mathcal{C}}\|\bm{A}_{ji}^{(\ell)}\|_F^2}}\right\}.
\]
\item \emph{Per-factor Frobenius budget}
$\|\bm{A}_{ji}^{(\ell)}\|_F^2\le P_{ji}^{(\ell)}$: each factor is
rescaled independently onto its own Frobenius ball.
\item \emph{Diagonal/scalar/unit-modulus controls}: the structural
form is enforced by directly parameterizing the underlying
scalar/vector entries (\cref{sec:param_feasible}); any additional
norm or modulus constraint is projected in this lower-dimensional
parameter space (e.g., entry-wise renormalization
$\theta_m\!\leftarrow\!\theta_m/|\theta_m|$ for unit modulus).
\end{itemize}
More elaborate projections (orthogonal/Stiefel manifolds, low-rank
truncation via SVD) are equally compatible with the framework but are
not used in the experiments of \cref{sec:experiments}.

\subsection{Input Shaping and Effective-Channel Capacity}
\label{sec:input_shaping}

We now extend the framework to input-covariance shaping. A simple
virtual-edge augmentation absorbs $\bm{\Sigma}_X$ optimization into
the same K-recursion + PGA pipeline (\cref{rem:input_cov}), and the
resulting maximum MI coincides with the capacity of the effective
vector Gaussian channel of \cref{rem:effective_channel}
(\cref{prop:capacity}).

\begin{remark}[Input covariance optimization via a virtual source]
\label{rem:input_cov}
The framework also accommodates optimization of the input covariance
itself (\cref{fig:topology}(d)).
Introduce a virtual source node $S\sim\CN(\bm{0},\bm{I})$ and a
controllable input factor $\bm{Q}$ such that $X=\bm{Q}S$.
Then $\bm{\Sigma}_X=\bm{Q}\bm{Q}^{\herm}$, and a total power
constraint $\tr(\bm{\Sigma}_X)\le P$ becomes $\|\bm{Q}\|_F^2\le P$,
which falls within the Frobenius-ball projections above.
The augmented DAG is handled by the same K-recursion
\eqref{eq:K_recursion} with
$\bm{K}_{SS}=\bm{I}$ and $\bm{K}_{XX}=\bm{Q}\bm{Q}^{\herm}$, so
input-covariance design reduces to the optimization of one
additional controllable edge factor.
\end{remark}

The virtual-edge reformulation of \cref{rem:input_cov} carries a
sharp information-theoretic consequence. By
\cref{rem:effective_channel}, any linear Gaussian DAG with fixed
non-source edge factors collapses to a single vector Gaussian channel,
whose capacity under a Gaussian-input power budget is attained by
classical water-filling. The following proposition records that the
maximum end-to-end MI achievable by controlling $\bm{Q}$ alone is
exactly this capacity, irrespective of the DAG topology.

\begin{proposition}[Effective-channel capacity]
\label{prop:capacity}
Consider a linear Gaussian DAG whose non-source edge factors are
fixed, with effective channel $Y=\bm{G}_M X+R_M$ and
$R_M\sim\CN(\bm{0},\bm{C}_{MM})$ (\cref{rem:effective_channel}), and
assume $\bm{C}_{MM}\succ\bm{0}$. Let $\{\lambda_i\}_{i=1}^{d_X}$ be
the (nonnegative) eigenvalues of
$\bm{G}_M^{\herm}\bm{C}_{MM}^{-1}\bm{G}_M$, and let $\mu>0$ be the
water level satisfying
$\sum_{i:\lambda_i>0}[\mu-1/\lambda_i]_{+}=P$, where the sum is
restricted to the positive eigenmodes (zero eigenvalues contribute
zero capacity). Under the virtual-edge construction of
\cref{rem:input_cov} with $\|\bm{Q}\|_F^2\le P$,
\begin{equation}
\sup_{\|\bm{Q}\|_F^2\le P}\MI(X;Y)
=\sum_{i:\lambda_i>0}\bigl[\log(\mu\lambda_i)\bigr]_{+},
\label{eq:capacity}
\end{equation}
which equals the capacity of the effective channel under the
Gaussian-input power constraint, attained by the water-filling
allocation $p_i^{\star}=[\mu-1/\lambda_i]_{+}$ on the positive
eigenmodes of $\bm{G}_M^{\herm}\bm{C}_{MM}^{-1}\bm{G}_M$.
\end{proposition}

\begin{proof}
$\|\bm{Q}\|_F^2=\tr(\bm{Q}\bm{Q}^{\herm})=\tr(\bm{\Sigma}_X)$, so the
Frobenius constraint maps to $\tr(\bm{\Sigma}_X)\le P$. By
\cref{rem:effective_channel},
\[
\MI(X;Y)=\log\det\bigl(\bm{I}+
\bm{C}_{MM}^{-1/2}\bm{G}_M\bm{\Sigma}_X\bm{G}_M^{\herm}
\bm{C}_{MM}^{-1/2}\bigr).
\]
Diagonalize $\bm{G}_M^{\herm}\bm{C}_{MM}^{-1}\bm{G}_M
=\bm{U}\diag(\lambda_1,\dots,\lambda_{d_X})\bm{U}^{\herm}$ with
$\lambda_i\geq 0$. By the standard MIMO water-filling
theorem~\cite{telatar1999}, or equivalently by applying Hadamard's
inequality after rotating the input covariance into this eigenbasis,
an optimal $\bm{\Sigma}_X$ can be taken diagonal in the eigenbasis,
$\bm{\Sigma}_X=\bm{U}\diag(p_1,\dots,p_{d_X})\bm{U}^{\herm}$, with
$p_i\geq 0$ and $\sum_i p_i\leq P$. The objective then factorizes
as $\MI=\sum_i\log(1+\lambda_i p_i)$. Modes with $\lambda_i=0$
contribute zero for any $p_i$, so the maximization reduces to
allocating the budget among the positive eigenmodes; classical
water-filling yields $p_i^{\star}=[\mu-1/\lambda_i]_{+}$ on those
modes and gives \eqref{eq:capacity}. In the degenerate case
$\lambda_i=0$ for all $i$, the effective channel carries no
information and $\MI(X;Y)=0$ for any feasible $\bm{Q}$, matching the
empty sum in \eqref{eq:capacity}.
\end{proof}

\Cref{prop:capacity} suggests two natural deployments of the
framework. \emph{Joint optimization}: run PGA simultaneously over all
controllable factors including the virtual edge $\bm{Q}$.
\emph{Two-stage decomposition (a pragmatic recipe)}: (Stage 1)
optimize the non-source edge factors by PGA with a fixed input
distribution (e.g., white $\bm{\Sigma}_X=\bm{I}$) to obtain the
effective channel $(\bm{G}_M,\bm{C}_{MM})$, then (Stage 2) solve
$\bm{Q}$ in closed form via \cref{prop:capacity}. This decomposition
is operationally simple and removes $\bm{Q}$ from the numerical
optimization, but is in general \emph{suboptimal} for the joint
problem: the non-source factors that maximize the MI under the
fixed Stage~1 input do not coincide with those that maximize the
effective-channel capacity realized by the Stage~2 output. Reaching
the joint maximum therefore requires either joint PGA on the
augmented parameter set, or alternating optimization between the two
stages; both procedures target the joint maximum but, like PGA in
general (\cref{sec:pga}), provide only stationary-point convergence.

\section{Numerical Experiments}
\label{sec:experiments}

We demonstrate the framework on four representative DAG classes
(\cref{sec:exp_setup,sec:exp_results}) plus a multi-layer
network as a direct illustration of the topology-agnostic claim
(\cref{sec:exp_random_network}). All experiments are optimized by
\emph{the same} PGA loop of \cref{sec:pga} with no topology-specific
code path: only the parent dictionary, the controllable/fixed edge
factorization, and the noise covariances change between settings. The
Wirtinger gradient $\nabla_{\param^{*}}\MI$ in
\eqref{eq:wirtinger_grad_factor} is computed by PyTorch's complex
reverse-mode automatic
differentiation~\cite{paszke2019pytorch}; no manual derivative is
implemented in any case.

\subsection{Topologies and Setup}
\label{sec:exp_setup}

\Cref{fig:topology} shows the four DAGs, with per-panel parameters
listed in \cref{tab:exp_setup}. Following the 1-indexed convention of
\cref{sec:nodes_edges} ($V_1=X$, $V_M=Y$), the four panels read:
\begin{itemize}[leftmargin=1.2em,itemsep=0pt,topsep=2pt]
\item[(a)] \emph{Single-link MIMO} ($M=2$): $V_2=\bm{H}\bm{F}V_1+\bm{Z}_2$;
\item[(b)] \emph{Diamond DAG} ($M=4$):
            $V_2=\bm{A}_{21}V_1+\bm{Z}_2$,
            $V_3=\bm{A}_{31}V_1+\bm{Z}_3$,
            $V_4=\bm{A}_{42}V_2+\bm{A}_{43}V_3+\bm{Z}_4$;
\item[(c)] \emph{Two-hop AF relay} ($M=3$):
            $V_2=\bm{H}_1V_1+\bm{Z}_2$,
            $V_3=\bm{H}_2\bm{R}V_2+\bm{Z}_3$;
\item[(d)] \emph{Virtual edge} ($M=3$):
            $V_2=\bm{Q}V_1+\bm{Z}_2$, $V_3=\bm{H}V_2+\bm{Z}_3$, with a
            virtual root $V_1=\widetilde{X}\sim\CN(\bm{0},\bm{I}_d)$.
\end{itemize}
The input $V_1=X$ is white circular Gaussian ($\bm{\Sigma}_X=\bm{I}_d$
in (a)--(c), $\bm{\Sigma}_{\widetilde{X}}=\bm{I}_d$ in (d)), the
$\bm{Z}_j\sim\CN(\bm{0},\sigma^2\bm{I})$ are independent with the
panel-specific $\sigma$ of \cref{tab:exp_setup}, and the controllable
factors are highlighted in red in \cref{fig:topology}. Panel~(d)
applies a small stabilization noise $\bm{\Sigma}_2=\varepsilon\bm{I}$
on $V_2$ with $\varepsilon=10^{-8}\ll\lambda_{\min}(\bm{Q}\bm{Q}^{\herm})$,
making $\tr(\bm{Q}\bm{Q}^{\herm})\le P$ coincide with $\|\bm{Q}\|_F^2\le P$
up to $O(\varepsilon)$.

Fixed channel matrices and the initial controllable factor are
sampled element-wise with independent unit-variance real and
imaginary parts (i.e., each entry is $\CN(0, 2)$); the fixed merging
matrices $\bm{A}_{42}$ and $\bm{A}_{43}$ in panel~(b) are
additionally rescaled by $0.5$, and the initial controllable factor
by $0.1$ so that the starting MI is small. As a
topology-agnostic reference, we report the uniform baseline obtained
by setting each controllable factor to the scaled identity that
saturates the power budget; the corresponding MI is the dashed
horizontal line in \cref{fig:pga_gallery}.

\begin{table}[t]
\centering
\caption{Per-panel parameters and feasible set. Controllable factors
are updated by PGA; fixed factors are sampled once per problem
instance and held constant across iterations.}
\label{tab:exp_setup}
\renewcommand{\arraystretch}{1.15}
\begin{tabular}{@{}lcccl@{}}
\toprule
Panel & $d$ & $\sigma$ & Feasible set & Controllable \\
\midrule
(a) MIMO        & 3 & 0.5 & $\|\bm{F}\|_F^2 \le 5$ & $\bm{F}$ \\
(b) Diamond     & 2 & 0.3 & $\|\bm{A}_{21}\|_F^2 + \|\bm{A}_{31}\|_F^2 \le 4$
                                                  & $\bm{A}_{21},\bm{A}_{31}$ \\
(c) AF relay    & 3 & 0.4 & $\|\bm{R}\|_F^2 \le 3$ & $\bm{R}$ \\
(d) Virtual edge & 3 & 0.5 & $\|\bm{Q}\|_F^2 \le 5$ & $\bm{Q}$ \\
\bottomrule
\end{tabular}
\end{table}

\begin{figure}[!t]
\centering
\includegraphics[width=\linewidth]{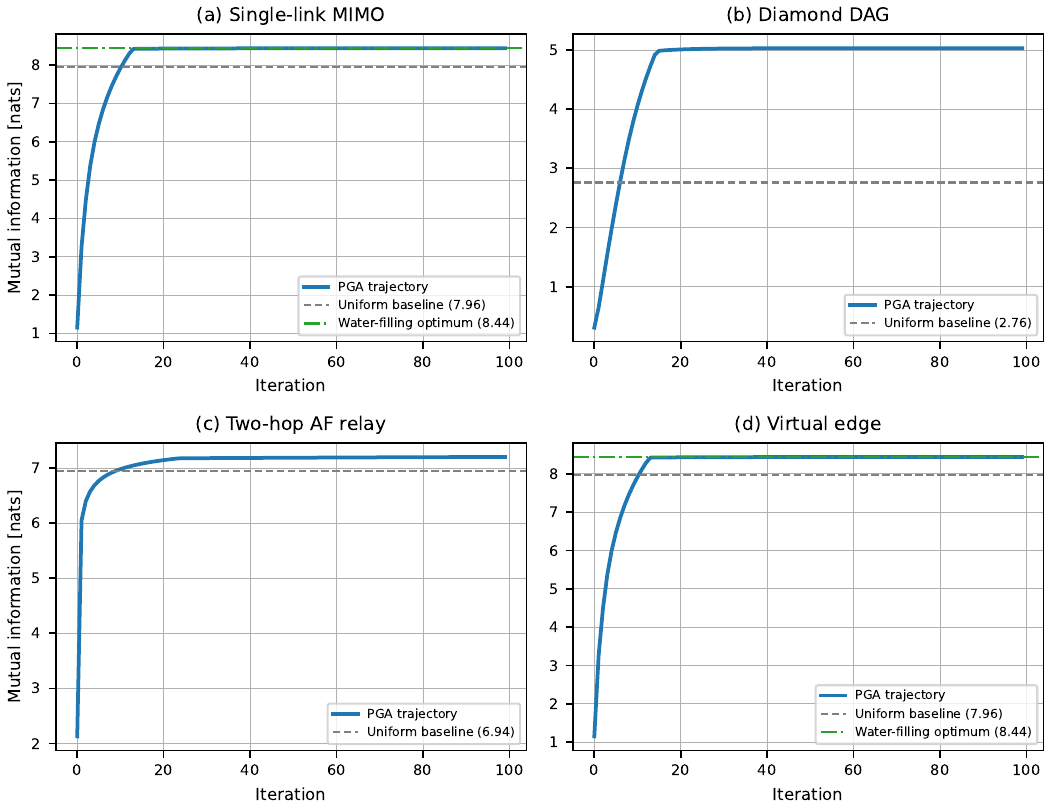}
\caption{PGA trajectories on the analytic information gradient, one
panel per DAG topology of \cref{fig:topology}. Dashed lines: uniform
baseline; dash-dotted lines in (a) and (d): water-filling MI.}
\label{fig:pga_gallery}
\end{figure}

\subsection{Results}
\label{sec:exp_results}

\Cref{fig:pga_gallery} shows the MI trajectory under PGA for one
representative problem instance per panel; all fixed matrices and
the initial controllable factor are sampled from a single random
seed as described above. PGA uses a constant step size
$\alpha=0.05$ and is run for a fixed budget of $T=100$ iterations,
without an early-stopping criterion; the choice of $\alpha$ and
$T$ is uniform across all four panels rather than tuned per panel,
and the panel-specific noise variance $\sigma$ (\cref{tab:exp_setup})
spans a representative range of SNR regimes.
For Panels~(a) and (d) we additionally plot, as a reference, the
classical water-filling MI on $\bm{H}$, which is well-defined
because $X$ is white and the panel reduces to a single-link MIMO
optimization in
$\bm{F}$ (resp.\ $\bm{Q}$).

\emph{(a)} Starting from a small random initialisation, PGA reaches
$8.4414$ nats, exceeding the uniform baseline $7.96$ and matching
the water-filling MI ($8.4419$ nats) to
$4.78\times 10^{-4}$ nats (relative error $\sim 6\times 10^{-5}$);
the gap shrinks to $10^{-7}$--$10^{-8}$ at 200--300 iterations. This
serves as an implicit correctness check that the generic
K-recursion + PGA pipeline reproduces the textbook optimum without
MIMO-specific code. As an explicit gradient-level check in the
single-link case, the reverse-mode AD gradient at the converged
precoder $\bm{F}^{\star}$ agrees with the Palomar--Verd\'u closed
form $\sigma^{-2}\bm{H}^{\herm}\bm{H}\bm{F}\bm{E}$~\cite{palomar2006}
to a relative Frobenius error below $10^{-14}$ in IEEE double
precision.
\emph{(b)} The shared root $V_1=X$ makes the cross-covariance
$\bm{K}_{32}=\bm{A}_{31}\bm{\Sigma}_X\bm{A}_{21}^{\herm}\neq\bm{0}$
and enters the merging-node block $\bm{K}_{44}$
(\cref{fig:K_dependency}); ignoring cross-blocks would yield an
incorrect $\bm{\Sigma}_Y$. Without any topology-specific gradient
derivation, PGA reaches $5.03$ nats, improving by $2.27$ nats over
the uniform branch-precoder baseline $2.76$.
\emph{(c)} The factor $\bm{R}$ enters $\bm{A}_{32}=\bm{H}_2\bm{R}$ and
also amplifies the relay-side noise, producing a nontrivial
signal--noise trade-off; PGA reaches $7.20$ nats, improving by
$0.26$ nats over the uniform relay-gain baseline $6.94$, without a
relay-specific gradient derivation.
\emph{(d)} Adding the virtual root and edge $\bm{Q}$ turns
input-covariance shaping into ordinary edge optimization; PGA again
reaches $8.4414$ nats, matching the water-filling MI of Panel~(a) to
the same accuracy and confirming the equivalence claimed in
\cref{sec:edge_factor}.

\subsection{Arbitrary Topology: A Multi-Layer Network}
\label{sec:exp_random_network}

\begin{figure}[!htbp]
\centering
\includegraphics[width=0.85\linewidth]{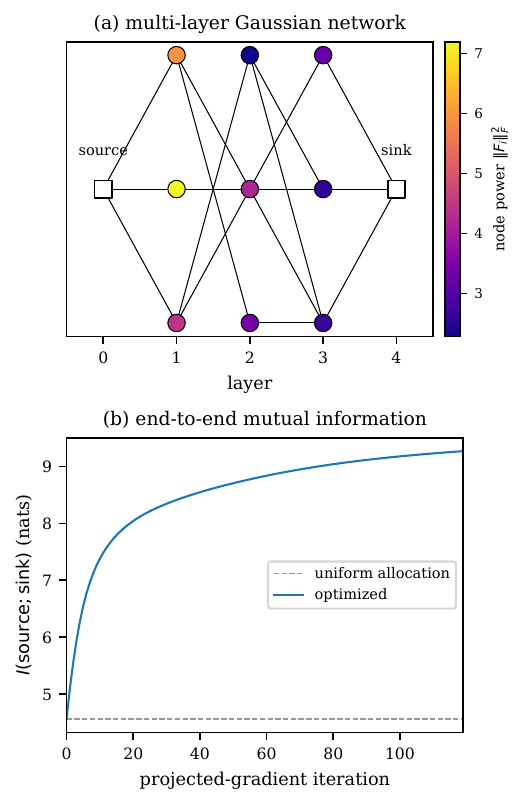}
\caption{Multi-layer Gaussian network ($M=11$
nodes, $L=5$ layers, $17$ edges, $d=4$). \emph{Top}: network topology
with relay nodes shaded by the optimized power
$\|\bm{F}_i^{\star}\|_F^2$ under the shared budget $P=36$.
\emph{Bottom}: end-to-end mutual information $I(X;Y)$ versus PGA
iteration.}
\label{fig:random_network}
\end{figure}

The four panels of \cref{fig:pga_gallery} all admit either a known
closed-form reference~\cite{telatar1999,palomar2006}
(Panels~(a),(d)) or a small, structured topology
amenable to analytical inspection (Panels~(b),(c)). To illustrate
extension beyond these small benchmark DAGs, we now run the same PGA
loop on a multi-layer Gaussian network for which no closed-form
capacity expression is available and the gradient of $I(X;Y)$ cannot
be hand-derived.

\paragraph*{Setup}
The DAG is a layered network with $M=11$ nodes spanning $L=5$ layers
(a source, $3$ relay layers of $W=3$ nodes each, and a sink) and
$17$ edges, with the topology shown in the top panel of
\cref{fig:random_network}. Each node carries a $d=4$-dimensional
complex Gaussian vector. The source $V_1=X$ emits an isotropic signal
$X\sim\CN(\bm{0},\bm{I}_d)$; each of the $9$ relay nodes carries a
controllable processing matrix $\bm{F}_i\in\mathbb{C}^{d\times d}$.
Concretely, each outgoing edge from relay~$i$ is factorized as
$\bm{A}_{ji}=\bm{H}_{ji}\bm{F}_i$, where $\bm{H}_{ji}$ is the fixed
channel and $\bm{F}_i$ is the relay's shared controllable processing
matrix (source-outgoing edges carry only $\bm{H}_{ji}$), realizing
the parameter-sharing pattern introduced in
\cref{sec:param_feasible}.
The framework jointly optimizes $\{\bm{F}_i\}_{i=1}^{9}$ under a
\emph{shared total power budget}
$\sum_{i=1}^{9}\|\bm{F}_i\|_F^2\le P$ with $P=36$, chosen so the
uniform initialization corresponds to identity processing at every
relay ($\|\bm{F}_i\|_F^2=d=4$ per relay). The same PGA loop of
\cref{sec:pga} with step size $\alpha=0.05$ is run for $T=120$
iterations.

\paragraph*{Results}
\Cref{fig:random_network} reports the outcome. The top panel
visualizes the network, with relay nodes shaded by their
optimized power $\|\bm{F}_i^{\star}\|_F^2$, revealing the discovered
network-wide allocation; the bottom panel plots the end-to-end mutual
information versus PGA iteration. From the uniform initialization
$I(X;Y)=4.56$ nats, the framework reaches $9.28$ nats, a
$2.0\times$ increase obtained by jointly tuning all $9$ relay matrices
under the shared budget. The optimized relay powers range from
$2.28$ to $7.19$ against the uniform share of $4.00$, demonstrating
that the framework redistributes the shared budget \emph{non-uniformly}
across relays; the total power
$\sum_i\|\bm{F}_i^{\star}\|_F^2=36$ is held throughout by the projection
step. No per-topology gradient was derived for this network: a single
K-recursion forward pass through the $M$-node graph followed by one
reverse-mode AD sweep produced the Wirtinger gradient at every relay
simultaneously, and the shared-budget projection distributed the power
network-wide.

\section{Conclusion}
\label{sec:conclusion}
We presented a unified differentiable framework with exact MI
evaluation and exact gradient computation for end-to-end mutual
information optimization over linear Gaussian DAGs, targeting
network-wide design under global constraints. The K-recursion
analytically constructs all node-pair covariances
(\cref{prop:K_recursion}) and the mutual information follows in closed
form (\cref{prop:MI_from_K}); the DAG structure simultaneously orders
this forward construction and the reverse-mode AD backward sweep,
delivering the exact Wirtinger gradient at every controllable factor
in a single reverse-mode backward sweep. A key strength is that no closed-form gradient
expression per topology (such as the Palomar--Verd\'u single-link MIMO
formula~\cite{palomar2006}) is required: the same pipeline applies
across linear Gaussian DAG topologies without topology-specific
gradient derivations. Projected gradient ascent is then used to maximize
the mutual information under global constraints such as a total
transmit power budget, with stationary-point convergence guaranteed
under standard step-size schedules. Numerical experiments on four representative DAG classes confirmed
that, through a single topology-agnostic implementation, the
framework recovered the classical water-filling optimum to numerical
precision in the cases where it is available and yielded MI
improvements in non-single-link representative topologies; a further
demonstration on a multi-layer network (11 nodes, 5 layers), for
which no closed-form gradient is available, illustrated that the
same loop extends to a nontrivial multi-layer topology and discovers
a non-uniform network-wide power allocation.
Future work includes stochastic objectives for fading channels
(developing fading-channel analogues of the present DAG-based
gradient framework, in the spirit of~\cite{palomar2006}) and
Bussgang-type extensions to nonlinear elements in the network.

\section*{Acknowledgment}
This work was supported by JST, CRONOS, Japan Grant Number JPMJCS25N5.


\end{document}